\documentclass[journal]{IEEEtran}
\usepackage{amssymb,amsfonts,color,fancybox}
\usepackage[mathscr]{eucal}
\usepackage{graphicx}
\usepackage{ifpdf}
\usepackage{longtable}	
\usepackage{multicol}
\usepackage{float}
\usepackage{algpseudocode}
\usepackage[sort,compress]{cite}
\usepackage{graphicx}
\usepackage{caption}
\usepackage{subcaption}

\graphicspath{{Figures/}}

\title{Solving the power flow equations: a monotone operator approach}

\author{Krishnamurthy Dvijotham, Steven Low, Michael Chertkov
\thanks{K.D.  and S.L are with the department of Computing and Mathematical Sciences, Caltech, M.C is with the T-Division, Los Alamos National Laboratory, Los Alamos, NM, 87544 USA. All queries should be addressed to K.D (e-mail: dvij@caltech.edu).}}

\usepackage{verbatim}
\usepackage{algorithm}
\usepackage{bm}
\usepackage{algpseudocode}
\usepackage[all]{xy}
\usepackage{url}
\usepackage{amsmath,amsthm,amssymb}
\usepackage{graphicx}
\usepackage{subcaption}

\newtheorem{theorem}{Theorem}[section]

\newtheorem{lemma}{lemma}
\theoremstyle{definition}
\newtheorem{definition}{Definition}

\theoremstyle{remark}
\newtheorem{remark}{Remark}
\theoremstyle{remark}
 
\usepackage{accents}
\newcommand{\PF}{\mathrm{PF}}
\newcommand{\Vcc}{x}
\newcommand{\distV}{d}
\newcommand{\Ly}[1]{\mathcal{L}_y\br{#1}}
\newcommand{\flim}{\gamma}
\newcommand{\Poly}[2]{\mathrm{Poly}^{#2}\br{#1}}
\newcommand{\normnuc}[1]{\left\|{#1}\right\|_\text{nuc}}

\newcommand{\SubM}[3]{{\left[{#1}\right]}^{#2}_{#3}}
\newcommand{\pqq}{\mathrm{pqq}}
\newcommand{\nf}{k}
\newcommand{\expc}[1]{\exp^c\br{#1}}

\newcommand{\E}{\mathcal{E}}
\newcommand{\G}{\mathrm{pv}}
\newcommand{\NSB}{\mathrm{nsb}}

\newcommand{\Vs}{\mathcal{V}}
\newcommand{\Submat}[3]{\left[{#1}\right]^{#2}_{#3}}

\newcommand{\Lo}{\mathrm{pq}}

\newcommand{\Vset}{v}

\newcommand{\Int}{\mathrm{Int}}

\newcommand{\Rep}[1]{\mathrm{Re}\br{#1}}
\newcommand{\Imp}[1]{\mathrm{Im}\br{#1}}
\newcommand{\Jac}[2]{\nabla {#1}\br{#2}}

\newcommand{\Vc}{V^{\log}}
\newcommand{\Vml}{\rho}
\newcommand{\Vpl}{\theta}

\newcommand{\PQ}{PQ}
\newcommand{\PV}{PV}
\newcommand{\herm}[1]{{#1}^\ast}

\newcommand{\Com}{\mathbb{C}}

\newcommand{\Sym}[1]{\mathrm{Sy}\br{#1}}

\newcommand{\Pa}{W}

\newcommand{\JF}{J_F}

\newcommand{\opt}[1]{{#1}^\ast}

\newcommand{\expb}[1]{\exp\left({#1}\right)}
\newcommand{\logb}[1]{\log\left({#1}\right)}
\newcommand{\norm}[1]{\left\lVert {#1} \right\rVert}

\newcommand{\One}{\mathbf{1}}
\newcommand{\tran}[1]{{#1}^T}

\newcommand{\inner}[2]{\left\langle {#1},{#2} \right\rangle}

\newcommand{\N}{\mathcal{N}}

\newcommand{\ic}{\mathbf{j}}

\newcommand{\detb}[1]{\det\left({#1}\right)}
\newcommand{\tranb}[1]{{\left({#1}\right)}^T}
\newcommand{\br}[1]{\left({#1}\right)}

\newcommand{\inv}[1]{{\left(#1\right)}^{-1}}

\newcommand{\diagb}[1]{\mathrm{di}\left({#1}\right)}

\newcommand{\tr}[1]{\mathrm{tr}\left( {#1} \right)}

\newcommand{\R}{\mathbb{R}}

\newcommand{\C}{\mathcal{D}}

\newcommand{\FHinf}[1]{q_\infty}
\newcommand{\FHtwo}[1]{q_2}
\newcommand{\FHone}[1]{q_1}

\begin{document}
	\maketitle
\begin{abstract}
	The AC power flow equations underlie all operational aspects of power systems. They are solved routinely in operational practice using the Newton-Raphson method and its variants. These methods work well given a good initial ``guess'' for the solution, which is always available in normal system operations. However, with the increase in levels of intermittent generation, the assumption of a good initial guess always being available is no longer valid. In this paper, we solve  this problem using the theory of monotone operators. We show that it is possible to compute (using an offline optimization) a ``monotonicity domain'' in the space of voltage phasors. Given this domain, there is a simple efficient algorithm that will either find a solution in the domain, or provably certify that no solutions exist in it. We validate the approach on several IEEE test cases and demonstrate that the offline optimization can be performed tractably and the computed ``monotonicity domain'' includes all practically relevant power flow solutions.
		\end{abstract}
\section{Introduction}

Power systems are experiencing revolutionary changes due to various factors, including integration of renewable generation, distributed generation, smart metering, direct or price-based load-control in operations. While potentially contributing to the long-term sustainability of the power grid, these developments also pose significant operational challenges by making the power system inherently stochastic and inhomogeneous. As these changes become more widespread, the system operators will no longer have the luxury of large positive and negative reserves. Thus, operating the future power grid will require developing new computational tools that can assess the system state and security margins more accurately and faster than current approaches. Specifically, these new techniques need to go beyond linear models and ensure that the power system is  safe even in the presence of large disturbances and uncertainty, where nonlinear effects dominate. In this paper, we focus on the fundamental equations of the power system : the power flow (PF) equations. The PF equations constitute a system of nonlinear equations and are known to exhibit complex and chaotic behavior \cite{araposthatis1981analysis}\cite{varaiya1992bifurcation}. Standard techniques like Newton-Raphson and its variants often fail to converge when the operating conditions are changing rapidly or the system is close to its security margins. In such a situation, it becomes difficult to assess whether the system is actually operationally unsafe or if the Newton-Raphson method failed because of numerical difficulties or bad initialization. In this paper, we propose an approach to remedy this problem. Our approach is based on the theory of monotone operators \cite{boydmonotone}. Just as a convex optimization problem can be solved efficiently, one can find zeros of a monotone operator efficiently (in fact, the former is a particular case of the latter as the gradient of a convex function is a monotone operator.) Thus, if we can show that the nonlinear PF equations can be described by a monotone operator over a sufficiently large domain, then they can be solved within the domain of monotonicity efficiently, or certify that no solution exists within the domain. . It turns out that the PF operator is not globally monotone, however it is monotone over a restricted domain.

Our main contribution is an efficient procedure based on semidefinite programming  to characterize an operationally relevant domain (that contains all the voltages satisfying operational constraints)  in the space of voltages over which the power flow operator is monotone. The domain is specified in terms of a simple bound on the voltage phasor ratios between neighboring buses. Roughly speaking, these bounds require that the voltage phasors between neighboring buses are ``sufficiently close''. Once this is done, there are well-known efficient algorithms to compute solutions of the PF equations satisfying these bounds, or certify that no solutions exist within the monotonicity domain. These algorithms are based on solving an associated \emph{monotone variational inequality}, for which several theoretically and practically efficient algorithms have been developed (see \cite{boydmonotone} and \cite{facchinei2007finite}, chapter 12).

As a by-product of our analysis, we obtain a domain over which the power flow Jacobian is non-singular. The domain is characterized by simple bounds on the voltage phasor ratios between neighboring buses. The bounds are interesting in their own right, as they allow the system operator to monitor distance to voltage collapse or loss of synchrony (where the Jacobian becomes singular) simply by monitoring ratios between voltage phasors on neighboring lines, which can even be done in a distributed manner. Numerical tests show that the bounds obtained are non-conservative and cover a wide range of operating conditions. 

The rest of this paper is organized as follows. Section \ref{sec:Intro} covers relevant background on power systems and monotone operators. The main technical results are presented in Section \ref{sec:Monotone}. In Section \ref{sec:Related}, we discuss how our approach compares to related work. In Section \ref{sec:Num}, we present numerical results illustrating our approach on some IEEE benchmark networks. 		
\section{Modeling Power Systems}\label{sec:Intro}

\newcommand{\neb}{\sim}
\subsection{Notation}
$\R$ is the set of real numbers, $\Com$ the set of complex numbers. $\R^n,\Com^n$ denote the corresponding Euclidean space in $n$ dimensions. Given a set $\C\subset\R^n$, $\Int\br{\C}$ denotes the interior of the set. Given a complex number $x\in\Com$, $\Rep{x}$ denotes its real part and $\Imp{x}$ its imaginary part. $\herm{x}$ denotes its complex conjugate. $\norm{x}$ refers to the Euclidean norm of a vector $x\in \R^n$ or $x\in\Com^n$ and $\inner{x}{y}$ to the standard Euclidean dot product. Given an vector $x\in\R^n$, $\diagb{x}$ denotes the $n\times n$ diagonal matrix with $\br{i,i}$-th entry equal to $x_i$.The \emph{nuclear norm} of $M$ is denoted by $\normnuc{M}$ and is equal to the sum of its singular values. Given a differentiable function $f:\R^k\mapsto\R^k$, $\nabla f$ denotes the Jacobian of $f$, a $k\times k$ matrix with the $i$-th row being the gradient of the $i$-th component of $f$. For $M\in\R^{n\times n}$, $\Sym{M}=\frac{M+\tran{M}}{2}$.
Given two sets of indices $S,S^\prime$ and matrix $M$ whose rows and columns are indexed by $S^\prime$, $N=\SubM{M}{S^\prime}{S}$ denotes the $|S|\times |S|$ matrix with the following property:
\begin{align*}
N_{ij} = \begin{cases}
M_{ij} & \text{ if } i,j\in S\cap S^\prime \\
0 & \text{ otherwise }
\end{cases}
\end{align*}
$\mathrm{tril}\br{M}$ denotes the vector formed by the lower triangular entries of matrix $M$.

\subsection{Background}
We represent the transmission network as a graph $\br{\Vs,\E}$ where $\Vs$ is the set of nodes and $\E$ is the set of edges. In power systems terminology, the nodes represent the buses and the edges correpond to power lines. Buses are denoted by indices $i=0,1,\ldots,n$ and lines by ordered pairs of nodes $\br{i,j}$. We pick an arbitrary orientation for each edge, so that for an edge between $i$ and $j$, only one of $\br{i,j}$ and $\br{j,i}$ is in $\E$. If there is an edge between buses $i$ and $j$, we write $i \sim j,j\sim i$. 

The transmission network is characterized by its complex admittance matrix $Y \in \Com^{n\times n}$. $Y$ is symmetric but not necessarily Hermitian. Let $G=\Rep{Y}$, $B=\Imp{Y}$.

Let $V_i$ be the voltage phasor, $p_i$ and $q_i$ denote active and reactive injection at the bus $i$ respectively. $V$ is the vector of voltage phasors at all buses. Three types of buses are considered in this work:
\underline{\PV~ buses} where active power injection and voltage magnitude are fixed, while voltage phase and reactive power are variables. The set of \PV~ buses is denoted by $\G$. The voltage magnitude set point at bus $i\in\G$ is denoted by $\Vset_i$.\\
\underline{\PQ~ buses} where active and reactive power injections are fixed, while voltage phase and magnitude are variables. The set of \PQ~ buses is denoted by $\Lo$.\\
\underline{Slack bus}, a reference bus at which the voltage magnitude and phase are fixed, and the active and reactive power injections are free variables.  We choose bus $0$ as the slack bus as a convention. The slack bus is voltage phasor by $V_0$.\\
We denote the union of \PV~and \PQ~buses as $\NSB=\G\cup\Lo$. 

We will work with the logarithmic-polar representation of the voltage phasor:
\begin{align*}
V_i=\expb{\Vml_i+\ic\Vpl_i},\Vml_i=\logb{|V_i|},\Vpl_i=\angle V_i	
\end{align*}
The variables in the power flow problem are the phases at the non-slack buses $\Vpl_{\NSB}$ and the voltage magnitudes at the \PQ~buses $\Vml_{\Lo}$. For brevity we write$\Vc=\begin{pmatrix}
	\Vpl_{\NSB} \\ \Vml_{\Lo}
\end{pmatrix}$	where $V=\expb{\Vml+\ic \Vpl}$. We also write $V=\expc{\Vc}$ and whenever we write $V$ it is understood that the constraints on $|V_{\G}|,V_0$ are satisfied. Let $\Vpl_{ij}=\Vpl_i-\Vpl_j,\Vml_{ij}=\Vml_i-\Vml_j$.
\begin{definition}[Power Flow Operator]
Define the power flow operator $F$ as
\begin{subequations}
\begin{align}
& [F\br{V}]_i = \sum_{j=0}^n B_{ij}\expb{\Vml_i+\Vml_j}\sin\br{\Vpl_{ij}} \nonumber \\ 
&+\sum_{j=0}^n G_{ij}\expb{\Vml_i+\Vml_j}\cos\br{\Vpl_{ij}}-p_i,1\leq i \leq n\label{eq:F2a}\\
&[F\br{V}]_{n+i} =\sum_{j=0}^nG_{ij}\expb{\Vml_i+\Vml_j}\sin\br{\Vpl_{ij}}\nonumber \\
&-\sum_{j=0}^n B_{ij}\expb{\Vml_i+\Vml_j}\cos\br{\Vpl_{ij}}-q_i,1\leq i \leq |\Lo| \label{eq:F2b}
\end{align}\label{eq:F}	
\end{subequations}
\begin{remark}
This is a simplified version of the practical power flow problem: We do not account for transformers (considering the entire network in renormalized voltage units) and we also do not consider more general $\pi$-model for power lines, accounting for shunt capacitors to the ground. However, all the aforementioned (and other) important practical details can be easily incorporated in the model we use and are not restrictive in terms of applications of our results to practical power systems. 	
\end{remark}
We denote by $\pqq=\Lo+n$ the set of indices of the \PQ~buses shifted by $n$. Thus, the power flow operator $F$ is indexed by $\NSB\cup\pqq$. The power flow equations can be written as $F\br{V}=0$ solved for $\Vc$. We denote by $\JF\br{V}$ the Jacobian of $F$ with respect to $\Vc$ evaluated at $V=\expc{\Vc}$. Note that this is not the standard power flow Jacobian in the power systems, since we differentiate wrt $\logb{|V|_{\Lo}}$ rather than $|V|_{\Lo}$.  We denote by $\nf$ the total number of variables being solved for ($\nf=|\NSB|+|\Lo|$).
\end{definition}
The next lemma expresses $\JF\br{V}$ as a quadratic function of $V$:
\begin{lemma}\label{lem:Jacob}
The power flow Jacobian $\JF\br{V}\in\R^{\nf \times \nf}$ can be written as a quadratic matrix function of the voltage phasors:
\begin{align}
\sum_{i \in \Lo}\Delta_i |V_i|^2+\sum_{\br{i,j}\in\E} \Gamma_{ij} \Rep{V_i\herm{V_j}}+\Psi_{ij} \Imp{V_i\herm{V_j}}
\end{align}
where 
\begin{align}
&\Delta_i = \Submat{\begin{pmatrix} 0 & G_i \\ 0 & B_i\end{pmatrix}}{\{i,n+i\}}{\NSB\cup\pqq} \\
&\Psi_{ij}=\Submat{\begin{pmatrix}
-G_{ij} & G_{ij} & B_{ij} & B_{ij} \\
-G_{ij} & G_{ij} & -B_{ij} & -B_{ij} \\
B_{ij} & -B_{ij} & G_{ij} & G_{ij}\\
B_{ij} & -B_{ij} & G_{ij} & G_{ij}
\end{pmatrix}}{\{i,j,n+i,n+j\}}{\NSB\cup\pqq} \\
&\Gamma_{ij}=\Submat{\begin{pmatrix}
B_{ij} & -B_{ij} & G_{ij} & G_{ij} \\
-B_{ij} & B_{ij} & G_{ij} & G_{ij} \\
G_{ij} & -G_{ij} & -B_{ij} & -B_{ij}\\
-G_{ij} & G_{ij} & -B_{ij} & -B_{ij}
\end{pmatrix}}{\{i,j,n+i,n+j\}}{\NSB\cup\pqq}
\end{align}
\end{lemma}
\begin{proof}
Via direct differentiation.
\end{proof}
\begin{remark}
The Jacobian need not be symmetric, as can be seen by the lack of symmetry in the matrices $\Delta_i,\Psi_{ij},\Gamma_{ij}$. However, each entry of the Jacobian matrix is a quadratic function of the voltage phasor $V$.
\end{remark}

\subsection{Monotone Operators}\label{sec:Mon}

We now review briefly the theory of monotone operators, as is relevant to the approach developed in this paper. For details and proofs of the results quoted in this section, we refer the reader to the recent survey \cite{boydmonotone}.
A function $H:\R^k\mapsto\R^k$ is said to be a monotone operator over a convex domain $\C$ if
\[\inner{H\br{x}-H\br{y}}{x-y}\geq 0 \quad \forall x,y\in\C\]
A monotone operator is a generalization of a monotonically increasing function (indeed, if $k=1$, the above condition is equivalent to monotone increase: $x\geq y \implies H\br{x}\geq H\br{y}$).
$H$ is said to be strictly monotone over $\C$ if
\[\inner{H\br{x}-H\br{y}}{x-y}>0 \quad \forall x,y\in\C,x\neq y\]
A common example of a monotone operator is the gradient of a differentiable convex function. 
\begin{definition}[Monotone Variational Inequality]
Let $\C\subset \R^k$ be a convex set and $H$ be a monotone operator over $\C$. The variational inequality (VI) problem associated with $H$ and $\C$ is:
\begin{align}
\text{ Find } x\in\C \text{ such that } \inner{H\br{x}}{y-x}\geq 0\quad \forall y\in\C\label{eq:VI}
\end{align}
Define the \emph{normal cone} to $\C$ at $x$ as $\N_\C\br{x}=\{y:\inner{y}{z-x}\leq 0 \forall z \in \C\}$. Then, the variational inequality is equivalent to finding $x\in\C$ such that $-H\br{x}\in\N_\C\br{x}$. 
\end{definition}

The following result shows that monotone variational inequalities with compact domains always have a solution and can be solved efficiently. 
\begin{theorem}\label{thm:VI}
If $H$ is strictly monotone operator over a compact domain $\C$, then \eqref{eq:VI} has a unique solution $\opt{x}$. Further, an approximate solution $x_\epsilon\in\C$ satisfying
\begin{align}
 \norm{x_\epsilon-x}\leq \epsilon\label{eq:VIapprox}
\end{align}
can be found using at most $O\br{\logb{\frac{1}{\epsilon}}}$ evaluations of $H$ and projections onto $\C$.
 \end{theorem}
 \begin{remark}
In this manuscript, we are interested in finding zeros of the PF operators introduced above. 
We can use monotone operator theory for this as follows: Suppose $H$ satisfies the hypotheses of theorem \ref{thm:VI}. If there exists a point $\opt{x}\in\C$ with $H\br{\opt{x}}=0$, then this is the unique solution of the variational inequality (figure \ref{fig:VIsolb}. Conversely, if the variational inequality has a solution with $H\br{\opt{x}}\neq 0$ (figure \ref{fig:VIsola}), then have a certificate that there is no solution of $H\br{x}=0$ with $x\in\C$.
\end{remark}
\begin{figure}
\centering
 \begin{subfigure}[b]{0.5\columnwidth}
       \includegraphics[width=\textwidth]{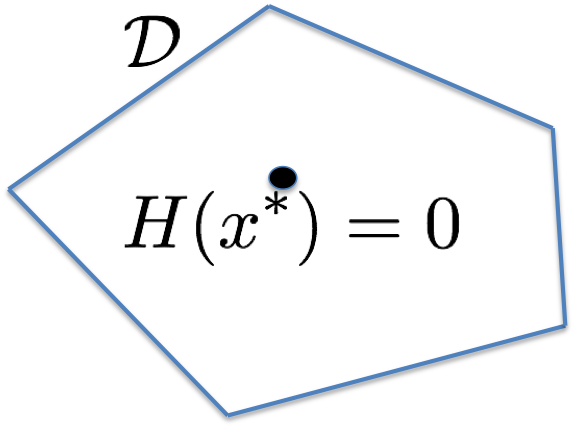}
       \caption{VI: solution in interior}
       \label{fig:VIsolb}
   \end{subfigure}
   \centering
 \begin{subfigure}[b]{0.8\columnwidth}
       \includegraphics[width=\textwidth]{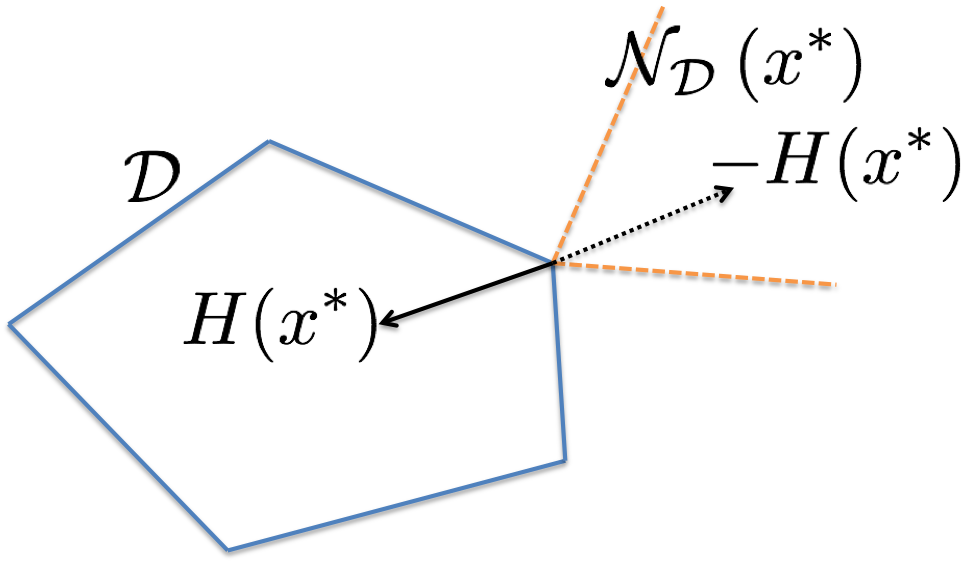}
       \caption{VI: solution on boundary}
       \label{fig:VIsola}
   \end{subfigure}
\end{figure}

The next result provides a simple characterization of monotonicity for differentiable operators:
 \begin{theorem}\label{thm:VIcond}
 Suppose that $H$ is differentiable. Then $H$ is strictly monotone over $\C$ if and only if
\[\Jac{H}{x}+\tran{\Jac{H}{x}}\succ 0 \quad \forall x \in \C\] 	
 \end{theorem}

\section{Monotonicity of the Power Flow Operator}\label{sec:Monotone}

In this section, we study the monotonicity of the PF operator $F$ \eqref{eq:F}. As described in Section \ref{sec:Mon}, zeros of $F$ (solutions to the PF equations) can be found efficiently if $F$ is monotone. Thus, if we can prove that the PF operator is monotone, the PF solutions can be found efficiently. Since PF equations can have multiple isolated solutions, it is not possible that the PF operator is globally monotone because this would imply a unique solution to the PF equations. Thus, we focus on characterizing domains over which the PF operator (or a scaled version of it) is monotone. This leads to the constrained power flow problem:
\begin{definition}[Constrained Power Flow Problem]\label{prob:ConstPF}
Given a set $\C \subset \R^{\nf}$, the constrained power flow problem is to determine whether the power flow equations $F\br{V}=0$ have a solution with $\Vc\in \C$, and if so, compute the solution. We denote this problem by $\PF\br{\C}$.
\end{definition}
\begin{remark}
In this paper, we will characterize sets $\C$ such that the constrained power flow problem can be solved. 	
\end{remark}


\subsection{Characterization of Domains of Monotonicity of the Power Flow Operator}

We now derive a procedure to characterize the domain of monotonicity of the PF operator \eqref{eq:F}. We first note that solving the equations $F\br{V}=0$ is equivalent to solving the equations $\Pa F\br{V}=0$ for an invertible matrix $\Pa \in\R^{\nf \times \nf}$. Define $F_{\Pa}\br{x}=\Pa F\br{\expc{x}}$ for brevity. 

\begin{theorem}\label{thm:Monotone}
Let $\C$ be any compact convex set in $\R^\nf$ such that $\exists \Pa \in \R^{\nf\times \nf}$ satisfying 
\begin{align}
\Sym{\Pa\JF\br{V}}\succ 0 \quad \forall V:\Vc\in \C\label{eq:MonCond}	
\end{align}
Then, $F_{\Pa}$ is strictly monotone over the set $\C$. Let $\opt{x}$ be the unique solution to the monotone variational inequality:
\begin{subequations}
\begin{align}
&\text{Find } x \in\C \nonumber\\
&\inner{F_{\Pa}\br{x}}{y-x}\geq 0 \quad\forall y\in\C \label{eq:PFVI}
\end{align}	
\end{subequations}
Then, if $F\br{\expc{\opt{x}}}=0$, $\expc{\opt{x}}$ is the unique solution of the PF equations in $\C$. Otherwise, there are no solutions of $F\br{V}=0,\Vc\in\C$. Thus, the constrained power flow problem can solved efficiently with domain $\C$.
\end{theorem}
\begin{proof}
This is a straightforward application of the theory of monotone operators. A similar proof can be found in \cite{DjMonPF}. 
\end{proof}

\begin{remark}
Theorem \ref{thm:Monotone} characterizes condition \eqref{eq:MonCond} under which the constrained power flow problem can be solved. In the reminder of this section, we show how one can construct sets $\C$ \emph{algorithmically} to satisfy this condition. Notice that \eqref{eq:MonCond} is an instance of a \emph{containment problem} that has been studied in the optimization literature \cite{kellner2013containment}.
\end{remark}

\subsection{Monotonicity Domains: 2-bus Network}
In order to motivate our choice of $\C$, we first consider a simple 2-bus network. Bus $0$ is the slack bus and bus $1$ is a \PQ~bus. Let $y=g-\ic b$ denote the complex admittance (conductance $g$, susceptance $b$) of the line between $0$ and $1$ and ignore any shunt elements. Let the slack bus voltage be $V_0=1$ (magnitude $1$, zero phase) and the voltage phasor at bus $1$ be $\expb{\rho+\ic\theta}$. The PF equations  are given by
\begin{align*}
p_1 & =b\expb{\rho}\sin\br{\theta}-g\expb{\rho}\cos\br{\theta}+g\expb{2\rho} \\
q_1 & =-g\expb{\rho}\sin\br{\theta}-b\expb{\rho}\cos\br{\theta}+b\expb{2\rho}	
\end{align*}
The power flow Jacobian (scaled by $\expb{-\rho}$) is given by
\begin{align*}
\begin{pmatrix}
b\cos\br{\theta}+g\sin\br{\theta} & 2g\expb{\rho}-g\cos\br{\theta}+b\sin\br{\theta} \\
-g\cos\br{\theta}+b\sin\br{\theta} & 2b\expb{\rho}-g\sin\br{\theta}-b\cos\br{\theta}
\end{pmatrix}	
\end{align*}
Choosing $\Pa=\frac{1}{b^2+g^2}\begin{pmatrix}b & -g \\ g & b\end{pmatrix}$, the scaled Jacobian becomes 
\begin{align*}
& \Pa\JF\br{\rho,\theta} = 	
& \expb{\rho}\begin{pmatrix}
	\cos\br{\theta} & \sin\br{\theta} \\
	\sin\br{\theta} & 2\expb{\rho}-\cos\br{\theta}
\end{pmatrix}
\end{align*}
The condition $\Sym{W\JF\br{\rho,\theta}}\succ 0$ reduces to the diagonal entries and the determinant being positive, which simplifies to: 
\begin{align}
\cos\br{\theta}>\frac{1}{2}\expb{-\rho}=\frac{1}{2}\frac{1}{|V|}.\label{eq:TwoBusCond}	
\end{align}
This is a well-known voltage stability criterion for the two-bus network \cite{PFIndex}. In this case, it is also easy to see that when $\cos\br{\theta}=\frac{1}{2|V|}$, the Jacobian is in fact singular, so that this is ``maximal'' monotonicity domain, in the sense that there is no monotonicity domain that contains it. In fact, the ``high voltage'' branch of the power flow solution set will lie always within this domain, until the point of maximum loadability is reached (beyond this point there are no solutions to the PF equations). 

This example also shows that the choice of $\Pa$ is critical. For example, if we choose $\Pa=I$, then we would automatically require $b\cos\br{\theta}+g\sin\br{\theta}>0$  (in addition to other conditions),leading to more restrictive constraints than \eqref{eq:TwoBusCond}.
\subsection{Algorithmic Computation of the Monotonicity Domain}\label{sec:MonCC}
For general networks, the situation is not as simple as in the case of the 2-bus network. There is no simple analytical choice of $\Pa$ that determines a large monotonicity domain. Instead, we propose a computational technique based on semidefinite programming. Motivated by the form of the 2-bus constraint \eqref{eq:TwoBusCond}, we consider domain $\C\br{\flim}$ defined as:
\begin{subequations}
\begin{align}
&\{\br{\Vpl_{\NSB},\Vml_{\Lo}}: 	\cos\br{\theta_{ij}}\geq \flim_{ij} \expb{|\rho_{ij}|}\, \forall \br{i,j}\in \E\} \\
&= \{\br{\Vc}: \Rep{V_i\herm{V_j}}\geq \flim_{ij} \max\br{|V_i|^2,|V_j|^2}\forall \br{i,j}\in \E\}
\end{align}\label{eq:SubEq}	
\end{subequations}
In the special case of the 2-bus network, the condition \eqref{eq:SubEq} reduces to $\cos\br{\theta}\geq \flim \expb{\rho}$, and as we saw before, one can choose $\flim=\frac{1}{2}$. The specific form of $\C\br{\flim}$ is convenient for our purposes and produces non-conservative estimates of the true monotonicity domain. Further, the condition reduces to simple bounds on the voltage phasors at the end of each transmission line, and may be useful in stability monitoring. However, depending on particular operational constraints at play in the system, we may consider other domains as well. 

\subsection{Computation of the Monotonicity Domain}\label{sec:MomRelax}
Certifying that $\C\br{\flim}$ is a monotonicity domain (i.e, checking that it satisfies condition \eqref{eq:CheckMonOpt}) amounts to checking that the following problem is feasible:
\begin{subequations}
\begin{align}
& \text{Find } \Pa \text{ such that } \nonumber\\
& \Sym{\Pa\JF\br{V}}\succ 0 \quad \forall V:\Vc\in\C 
\end{align}	\label{eq:CheckMonOpt}
\end{subequations}
Rewriting this explicitly, we need to solve the following optimization problem:
\begin{subequations}
\begin{align}
\max_{\Pa} \min_{z\in \R^{\nf},V \in \Com^n} \quad & \tran{z}\Sym{\Pa\JF\br{V}}z \\
\text{ Subject to } & \Rep{V_i\herm{V_j}}\geq \flim_{ij}\max\br{|V_i|^2,|V_j|^2} \\
					& |V_i|^2	=	\Vset_i^2, i \in \G \\
					& V_0=1 \\ 
					& \tran{z}z	=	1 \\
					& \normnuc{\Pa} \leq 1
\end{align}		
\end{subequations}

The constraint $\normnuc{\Pa}\leq 1$ is an arbitrary scaling constraint (since the problem is homogeneous in $\Pa$) where $\norm{\cdot}_*$ denotes the nuclear norm, or the sum of singular values of $\Pa$. If the optimal value of the above problem is positive, the optimal solution $\opt{\Pa}$ is such that $F_{\opt{\Pa}}$ is strictly monotone over $\C$. The inner minimization is a nonconvex quartic optimization problem in $\br{z,V}$ and is NP-hard in general. We relax the inner optimization problem using the moment-relaxation approach \cite{lasserre2009moments}. Let $\Vcc=\begin{pmatrix} 1 & \tran{z} & \tran{\Rep{V}} & \tran{\Imp{V}}\end{pmatrix}$. Let $\Poly{\Vcc}{i}$ denote the vector of all the monomials of degree upto $i$  in $\br{\Vcc}$ (with the first entry equal to the 0-degree monomial $1$). For example 
$\Poly{\Vcc}{2}=\begin{pmatrix}
	1 & \Vcc_1 & \Vcc_2 & \ldots & \Vcc_1^2 & \Vcc_1\Vcc_2 \ldots & 
\end{pmatrix}$.

Let $m$ be the size of this $\Poly{\Vcc}{4}$. We define a moment vector $y$ of the same size as $\Poly{\Vcc}{4}$ and the linear operator on the space of all degree-$4$ polynomials in $\br{\Vcc}$:
$\mathcal{L}_y\br{\sum_{i=1}^m c_i\mathrm{Poly}^4_i\br{\Vcc}}=\sum_{i=1}^m c_iy_i$.

We also define the \emph{localizing matrices}\cite{lasserre2009moments} (for $i=1,2$):
$X_{il} = \Poly{\Vcc}{i}\tranb{\Poly{\Vcc}{i}}$.
\begin{align*}
&\max_{\Pa} \min_{y} \quad  \tr{\Pa\Ly{\JF\br{V}z\tran{z}}} \\
&\text{ Subject to }  \\
& \Ly{\br{\Rep{V_i\herm{V_j}}-\flim_{ij}|V_i|^2}X_{1l}}\succeq 0\\
& \Ly{\br{\Rep{V_i\herm{V_j}}-\flim_{ij}|V_j|^2}X_{1l}}\succeq 0\\
& \Ly{\br{|V_i|^2-\Vset_i^2}X_{1l}}= 0, i \in \G \\
& \Ly{\br{|V_0-1|^2}X_{1l}}= 0 \\ 
& \Ly{\br{\tran{z}z-1}X_{1l}}=0 \\
& \Ly{X_{2l}}\succeq 0\\
& y_1=1,\normnuc{\Pa} \leq 1
\end{align*}	
This is a convex-concave saddle point problem with compact feasible sets. Thus, we have strong duality, we can switch the $\min$ and $\max$ and reduce the problem to:
\begin{subequations}
\begin{align}
& \min_{y,t} \quad  t\\
&\text{ Subject to }  \\
& \Sym{\Ly{\JF\br{V}z\tran{z}}}\preceq tI \label{eq:NegPSD}\\
& \Ly{\br{\Rep{V_i\herm{V_j}}-\flim_{ij}|V_i|^2}X_{1l}}\succeq 0\\
& \Ly{\br{\Rep{V_i\herm{V_j}}-\flim_{ij}|V_j|^2}X_{1l}}\succeq 0\\
& \Ly{\br{|V_i|^2-\Vset_i^2}X_{1l}}= 0, i \in \G \\
& \Ly{\br{|V_0-1|^2}X_{1l}}= 0 \\ 
& \Ly{\br{\tran{z}z-1}X_{1l}}=0 \\
& \Ly{X_{2l}}\succeq 0 \\
& y_1=1
\end{align}	\label{eq:MomentRelax}	
\end{subequations}
If the above problem has a positive optimal value, then we obtain a certificate for \eqref{eq:MonCond}. 
\eqref{eq:NegPSD} is a semidefinite program in $y$ that can be solved efficiently using interior point methods \cite{boyd2004convex}. $\Pa$ is given by the dual variable corresponding to the constraint \eqref{eq:NegPSD}. 

In section \ref{sec:CompTime}, we discuss ideas on scaling the approach to large networks by exploiting sparsity, and list the sizes of the resulting semidefinite programming problems.

\section{Discussion}\label{sec:Related}
The AC power flow equations are fundamental to all aspects of power systems operations and planning, and several decades of research have gone into developing algorithms for solving the power flow problem. A discussion of the Newton's method (by far the most popular algorithm developed for solving the PF equations) and its variants  can be found in standard textbooks on power engineering \cite{bergen2000power}. The idea here is to update the power flow variables according to:
\begin{align*}
&\Vc\br{i+1}\gets \Vc\br{i}-\eta_t\inv{\JF\br{\Vc\br{i}}}F\br{\br{\Vc\br{i}}^t}
\end{align*}
For a small enough step-size $\eta_t=\eta$, if the inverse Jacobian remains bounded ($\norm{\inv{\JF\br{\Vc}}}\leq \kappa$ for some $\kappa>0$), then the algorithm will converge to a power flow solution. Given a good initial guess for $\Vc$, Newton's method converges rapidly. However, if  the Jacobian becomes close to singular, then Newton's method can behave badly. In general, Newton's method can exhibit very complicated behavior and have fractal basins of attraction \cite{epureanu1998fractal}. Several approaches (damping, trust region methods etc.) have been proposed and studied to improve the stability and convergence of Newton's method.

In the context of power systems, optimal multiplier methods were proposed \cite{schaffer1988nondiverging}\cite{iwamoto1981load} to prevent divergence of Newton's method. These adapt the choice of step length in Newton's method to ensure that the algorithm does not diverge, although it may not converge to a power flow solution even if the solution exists. A recent survey of the optimal multiplier algorithms along with numerical comparisons can be found in \cite{tate2005comparison}.

An alternative approach to solving the PF equations is based on homotopy, where the power flow problem is first solved for an ``easy'' injection vector $p^0,q^0$, and the injections are changed gradually while tracking the power flow solution until the actual injections $p,q$ are reached. This approach is called continuation power flow in the context of power systems \cite{ajjarapu1992continuation}. It was initially claimed that this algorithm is capable of finding all power flow solutions, but this claim was shown to be incorrect recently \cite{molzahn2013counterexample}. Numerically this approach has been shown to be effective \cite{jin2005summarization}, although theoretical guarantees are still lacking , to the best of our knowledge.

Recently, a new approach based on complex analytic continuation techniques \cite{trias2012holomorphic} has generated a lot of interest and been shown to be effective for practical power systems problems. However, again, theoretical analysis of the conditions under which it is guaranteed to work are yet to be established..

A rigorous approach combining homotopy and tools from algebraic geometry was proposed \cite{mehta2014numerical} to find all power flow solutions. However, computational scalability of this approach is currently limited to small networks.

In \cite{madaniconvexification}, the authors propose a semidefinite programming (SDP) relaxation approach to solve the power flow equations. The authors characterize the set of voltage vectors such that the SDP has a rank-1 solution and hence recovers a physically valid power flow solution. This result is similar in spirit to ours. However, the set of recoverable voltages is specified by a nonconvex nonlinear matrix inequality which does not have a simple interpretation in terms of operational constraints on voltages. An exact comparison of the relative power of our approach and the approach proposed in this paper is a topic of future work.

The closest works to what is described in the paper are \cite{DjEnergyFun} and \cite{DjMonPF}. In \cite{DjEnergyFun}, we studied the case of losses power networks $\Rep{Y_{ij}}=0$ for every $\br{i,j}\in\E$.  In \cite{DjEnergyFun}, we have shown that in the case of lossless networks, the power flow equations can be solved by solving a convex optimization problem: Minimization of the energy function. The results of \cite{DjEnergyFun} are  a special case of the results in this paper. For lossless networks, the power flow operator $F$ is the gradient of the energy function for the structure preserving model of power systems dynamics \cite{varaiya1985direct} and $\JF$ is the Hessian. In this case, the domain of monotonicity of the power flow operator coincides with the domain of convexity of the energy function. We obtained a sufficient condition, expressed as a nonlinear but convex matrix inequality in $\Vc$, characterizing the set over which the energy function is convex (or equivalently, the power flow operator is monotone). However, for lossy networks, the monotonicity domain can no longer be characterized analytically, which led us to the computational procedure described in section \ref{sec:MomRelax}.

In \cite{DjMonPF}, we described an alternate approach to computing monotonicity domains. The main difference between the work here and that presented in \cite{DjMonPF} is the choice of coordinates used to compute the power flow Jacobian. In this paper, we worked with the log-polar coordinates $V_i=\expb{\Vml_i+\ic\Vpl_i}$. The alternative , explored in \cite{DjMonPF},  is to work with cartesian coordinates $V_i=V_i^x+\ic V_i^y$. The choice of coordinates has been a subject of several numerical studies. In our experience, using log-polar coordinates has two main advantages:\\
1) The size of the Jacobian is smaller, since the voltage magnitudes at \PV~buses are fixed and do not need to be solved for. This significantly speeds up the computational procedure described in \ref{sec:Mon}.\\
2) Our experiments show that the size of the monotonicity domain is significantly larger. The reason for this phenomenon is not clear yet, and will be a topic of future investigation.
\section{Numerical Results}\label{sec:Num}

In this section, we numerically validate our results. Section \ref{sec:MonExtent} describes the monotonicity domains computed for various test networks. Section \ref{sec:CompTime} discusses the sizes of the semidefinite programming problems produced while solving \eqref{eq:MomentRelax} and section \ref{sec:CompNewton} compares our approach to Newton's method and the default solver from matpower \cite{zimmerman2011matpower}, which is a variant of Newton's method with some step-size control.

\subsection{Extent of Monotonicity Domains}\label{sec:MonExtent}

We first consider the 3-bus network plotted in figure (\ref{fig:3busnetowrk}). All three buses are \PV~buses with voltage magnitudes set at $1$ p.u. In figure (\ref{fig:3busmondomain}), the blue region is the closed region whose boundary is defined by $\detb{\JF\br{V}}=0$. Any convex domain of strict monotonicity domain cannot intersect the boundary of this region, since at the boundary the Jacobian is singular which means that $\Sym{\Pa\JF\br{V}}\not\succ 0$ for any $\Pa$. Thus, any convex domain of strict monotonicity that contains the point $V=\One$ (all voltages equal to 1 p.u with $0$ phase, the zero-load PF solution) must be contained inside the blue region. Our approach computes a bound of $\flim=.08$, which defines the red region. As one can see from the figure, this covers almost the entire blue region, so that our estimate of the monotonicity domain is nearly tight.
\begin{figure}[htb]
        \centering
        \begin{subfigure}[b]{0.6\columnwidth}
                \includegraphics[width=\textwidth]{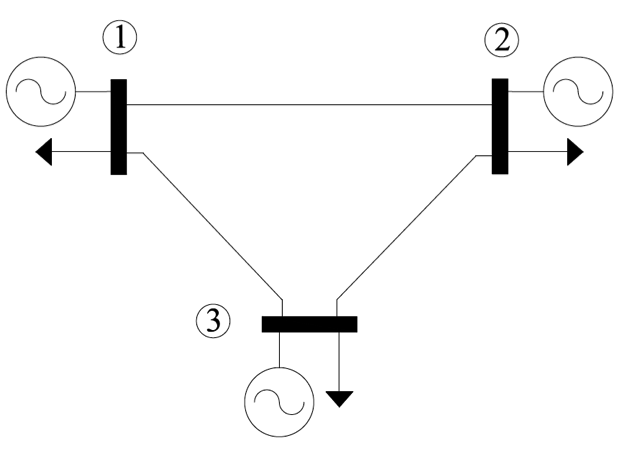}
                \caption{3 bus network}
                \label{fig:3busnetowrk}
        \end{subfigure}
        \begin{subfigure}[b]{0.7\columnwidth}
                \includegraphics[width=\textwidth]{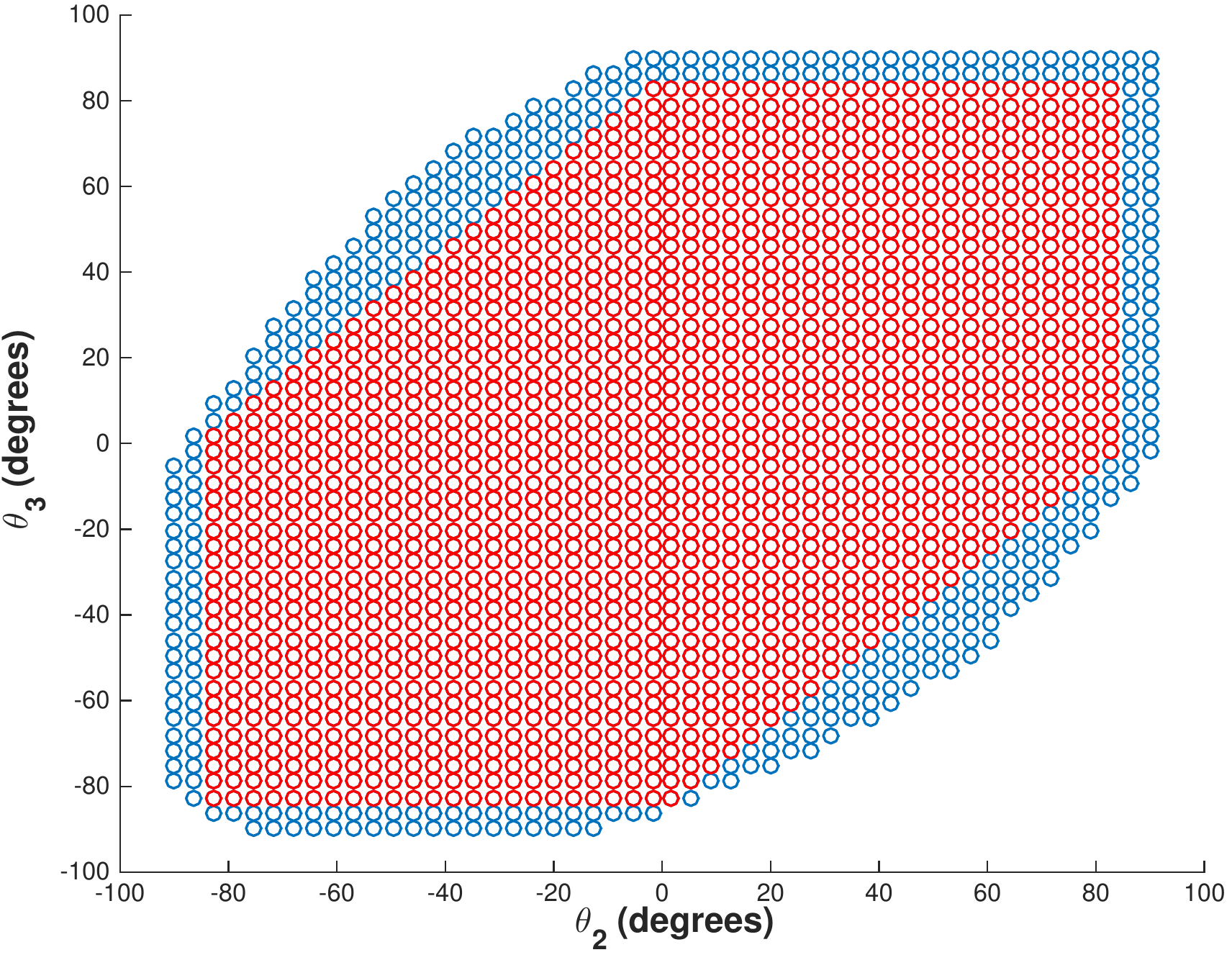}
                \caption{3 bus monotonicity domain. Blue Region: True Monotonicity Domain. Red Region: Estimated Monotonicity Domain}
                \label{fig:3busmondomain}
        \end{subfigure}
\end{figure}	
\begin{figure}[htb]
\centering
\includegraphics[width=.7\columnwidth]{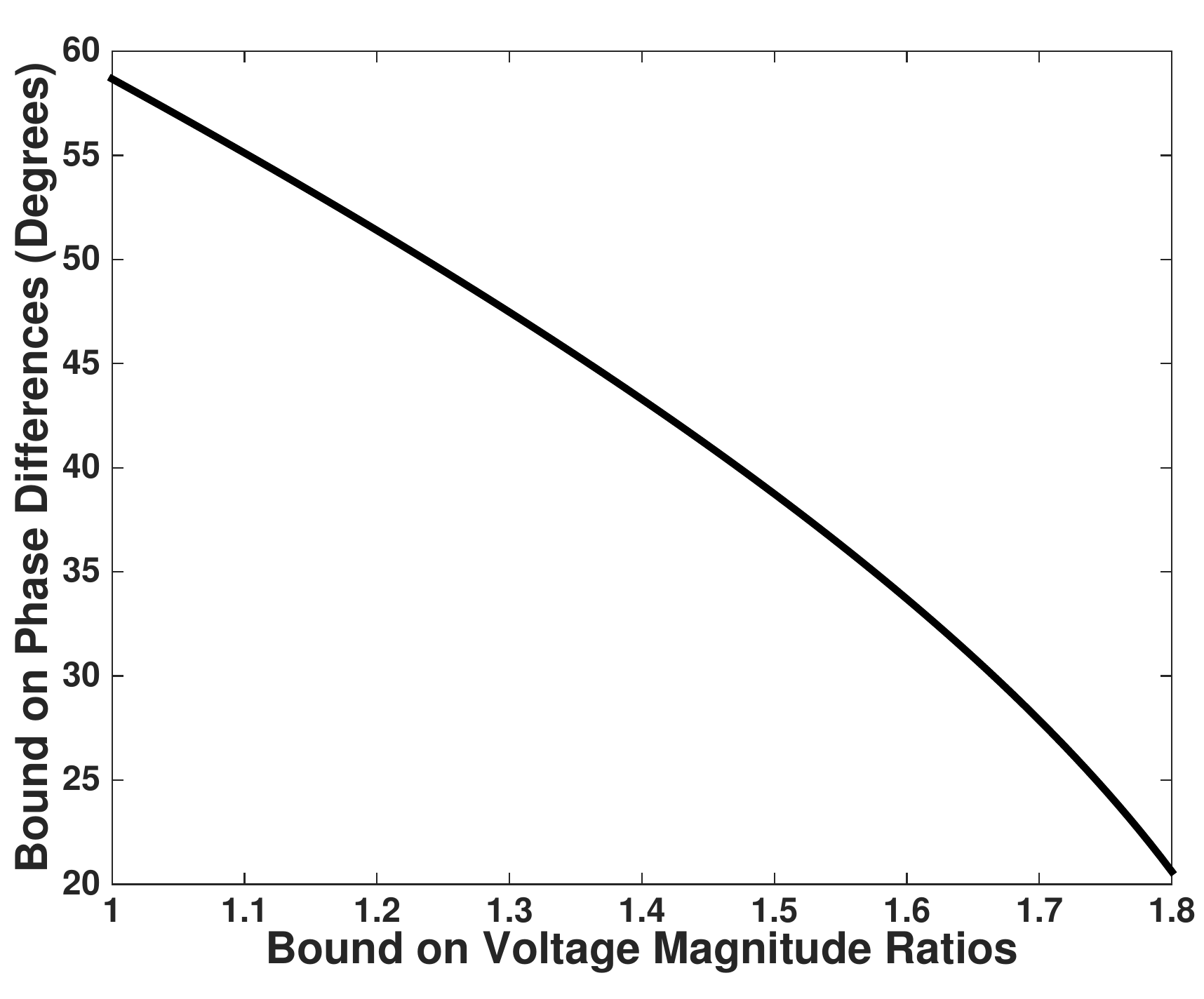}
                \caption{39 bus network: Tradeoff between bound on voltage ratios and phase differences}
                \label{fig:39bus}
\end{figure}

For larger networks, we find the smallest uniform value $\flim>0$ such that the optimal value of \eqref{eq:MomentRelax} is positive with $\flim_{ij}=\flim$. The values are listed in table \ref{tab:MinLim}. The first row lists the case number (matpower case file), the second row the minimum value of $\flim$ for the network (as a uniform bound on all transmission lines) and the third row turns $\flim$ into a bound on phase differences, assuming that the voltage magnitudes can vary between $.9$ and $1.1$ p.u at all \PQ~buses. The number in the third row is a simplification of the actual constraint (which couples voltage magnitudes and phases). More generally, there is a tradeoff between the bound on voltage magnitudes ratios and phase differences. We plot this tradeoff for the 39 bus network in figure (\ref{fig:39bus}). This shows that as voltage magnitudes are allowed to fluctuate more, phase differences need to kept within smaller limits to remain within the monotonicity domain.
\begin{table}
\begin{center}
\begin{tabular}{|c|c|c|c|c|}
\hline
Case Number & 9 & 14 & 30 & 39 \\
\hline
 Min $\flim$ & .31 & .34 & .41 & .52 \\
 \hline 
 Max $|\theta_i-\theta_j|$ & $67.7 \deg$ & $65.4 \deg$ & $59.9 \deg	$ & $50.54 \deg$ \\
\hline
\end{tabular}		
\end{center}
\caption{Minimum $\flim$ for different matpower test cases}	\label{tab:MinLim}
\end{table}

\subsection{Computation Time}\label{sec:CompTime}
Solving the moment relaxation \eqref{eq:MomentRelax} is the most expensive part of this approach. However, the situation is made better by the following factors:\\
1): This computation only depends on the network parameters and not on the specific injection profiles. Thus, the computation can be done offline as long as the network topology is fixed. Further, even if the network topology changes (because of transmission switching, loss of a generator/line etc.), one can do  offline computations for a large set of topology scenarios.\\
2) The problem \eqref{eq:MomentRelax} has a lot of sparsity. The Jacobian inherits the sparsity of the network, and each constraint only involves a small number of variables. 
One can exploit the sparsity in the moment relaxation to get more tractable Semidefinite Programs. Without sparsity, the size of the SDP would grow as $n^2$, where $n$ is the number of buses in the network, which becomes intractable for $n$ beyond $6$ or so. However, by exploiting sparsity using the ideas described in \cite{lasserre2006convergent}, we can reduce the size of the resulting SDPs significantly. We list the size of the largest PSD constraint for all the IEEE cases we tested in table \ref{tab:SizPSD}. One can see that the size grows gracefully.

With the above considerations, we can solve the 39 bus network certification problem in about 20 minutes on a MacBookPro 2.6 GHz Intel Core i7 notebook. We believe that by using the following ideas, the approach can be scaled to networks with several thousand buses:\\
1) In recent work \cite{molzahnLatest}\cite{molzahn2014sparsity}, the authors have shown that moment relaxations of the OPF problem can be solved for the networks with several thousand buses by applying higher order moment relaxations selectively. \\
2) It is possible to consider more restrictive conditions than positive definiteness of $\Sym{\Pa\JF\br{V}}$. For example, one can require that this matrix is diagonally dominant or generalized diagonally dominant \cite{ahmadi2014dsos}. This would give potentially more conservative results, but allow the SDP constraint to be replaced by an LP (linear program) or SOCP (second order cone program), thus giving a more scalable approach. \\
3) In the literature on graphical models and belief propagation \cite{sontag2010approximate} has studied LP relaxations of sparse integer programs. By discretizing the space of voltages, similar techniques may be applied the problems studied here as well. 
\begin{table}
\begin{center}
\begin{tabular}{|c|c|c|c|c|c|}
\hline
Case number & 9 & 14 & 30 & 39 & 57\\
\hline
Size of SDP & 78 & 78 & 105 & 136 & 190\\
\hline  
\end{tabular}	
\caption{Size of PSD Constraints in sparse moment relaxation}	\label{tab:SizPSD}
\end{center}
\end{table}

\subsection{Comparison to Newton's Method}\label{sec:CompNewton}
We also compare the performance of our method to a naive implementation of Newton-Raphson and the default matpower power flow solver \cite{zimmerman2011matpower}. For a given network, we generate a random voltage phasor vector with voltages at the \PQ~buses sampled uniformly from the interval $(.9,1.1)$ and phases at the non-slack buses sampled uniformly from the interval $(-\overline{\theta},\overline{\theta})$ where $\overline{\theta}$ is a bound on the absolute value of the voltage phase at each non-slack bus. We then form the injection vector corresponding to the voltage phasor and solve the resulting power flow problem using 3 methods: a) The monotone operator method described in this paper, b) A naive implementation of Newton-Raphson with a constant unit step-size, initialized with all voltage phases equal to $0$ and voltage magnitudes at \PQ~buses equal to $1$ p.u, c) Matpower's default PF solver. 
We generate a 100 random instances for each value of $\overline{\theta}$ and check whether each solver succeeded in finding a power flow solution. We then plot the probability of success of each solver as a function of $\overline{\theta}$. The results for the 9-bus, IEEE 14-bus and IEEE 39-bus network (taking data from the matpower case files) are plotted in figure \ref{fig:Comparison}. The results show that the monotone operator method is superior to the others. Further, it has the advantage of certifying the non-existence of a solution in the operationally relevant domain of voltage phasors, when it fails, a guarantee the other methods cannot provide.
\begin{figure}
   \centering
   \begin{subfigure}[b]{0.8\columnwidth}
       \includegraphics[width=\textwidth]{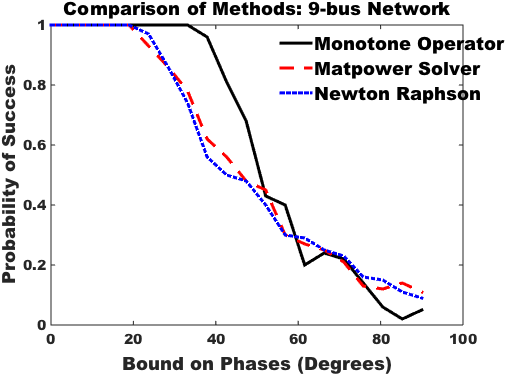}
       \caption{Comparison of Methods: 9-bus Network}
       \label{fig:9bus}
   \end{subfigure}
   ~ 
   \begin{subfigure}[b]{0.8\columnwidth}
       \includegraphics[width=.9\columnwidth]{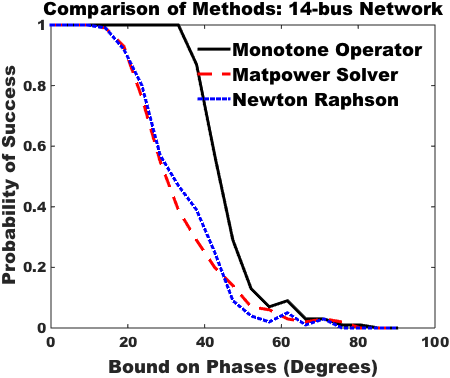}
       \caption{Comparison of Methods: 14-bus Network}
       \label{fig:14bus}
   \end{subfigure}
   ~ 
   \begin{subfigure}[b]{0.8\columnwidth}
       \includegraphics[width=\textwidth]{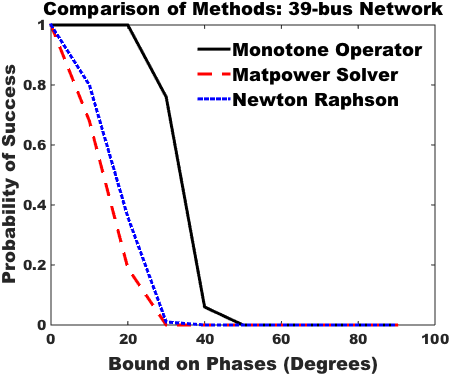}
       \caption{Comparison of Methods: 39-bus Network}
       \label{fig:39bus}
   \end{subfigure}
   \caption{Comparison of Monotone Operator methods with Newton-Raphson and Matpower's default solver}\label{fig:Comparison}
\end{figure}

\section{Conclusions}
We have developed a new approach to solving the power flow equations. The main feature of the new approach is its ability to find efficiently a solution within the ``monotonicity domain'' or prove that no solutions exist inside this domain. Our numerical results show that the monotonicity domain typically contains all power flow solutions satisfying standard operational constraints on voltage magnitudes and phases. The computation of the monotonicity domain can be done offline and exploits the moment relaxation approach, which has recently been successfully applied to large scale optimal power flow problems \cite{molzahn2014sparsity}\cite{molzahnLatest}.

Future work will focus on computational scalability of \eqref{eq:MomentRelax} and extensions of this work to further applications: state estimation, optimal power flow and small-signal stability analysis. In this regard, a related paper \cite{DjKostya} has recently been submitted, focusing on characterizing subsets of the feasible injection space in an OPF problem.

\section*{Acknowledgements}
We thank Dan Molzahn for extensive discussions and sharing code on sparse moment relaxations. The work at LANL was carried out under the auspices of the National Nuclear Security Administration of the U.S. Department of Energy at Los Alamos National Laboratory under Contract No. DE-AC52-06NA25396 and it was partially supported by DTRA Basic Research Project \#1002713399. The authors also acknowledge partial support of the Advanced
Grid Modeling Program in the US Department of Energy Office of Electricity.
\bibliographystyle{ieeepes}
\bibliography{Ref}	
\end{document}